\documentclass[10pt]{article}

\usepackage[T1]{fontenc}
\usepackage{graphicx}
\usepackage{booktabs}
\usepackage{tabularx}
\usepackage{longtable}
\usepackage{mathtools}
\usepackage[ruled,vlined]{algorithm2e}
\usepackage{makecell}

\usepackage{comment}
\usepackage{diagbox}
\usepackage{enumitem}
\usepackage[justification=centering]{caption}
\usepackage{amsmath,amsfonts}
\usepackage{amsthm}
\usepackage{url}
\usepackage{cleveref}
\usepackage{geometry}
\usepackage{color,xcolor,soul}

\newcommand{\Fq}{\mathbb{F}_q}
\newcommand{\ZZ}{\mathbb{Z}_4}
\newcommand{\C}{\mathcal{C}}

\newtheorem{theorem}{Theorem}[section]
\newtheorem{definition}[theorem]{Definition}
\newtheorem{proposition}[theorem]{Proposition}

\begin{document}

\title{New Codes from Cyclic and Negacyclic Codes of Even Length over $\mathbb{Z}_4$
\thanks{This research was supported by Kenyon College Summer Science Scholars
research program. Nuh Aydin is the corresponding author.}}

\author{Nuh Aydin \and Mohamed O. Belghith \and Godwin Idowu \and
Trang T.T. Nguyen \and Long B. Tran\\
Kenyon College, Gambier OH 43022, USA}

\date{}

\maketitle

\begin{abstract}
This paper uses theoretical results previously established in the literature to design search algorithms to find new linear codes over $\mathbb{Z}_4$ from cyclic and negacyclic codes of even length. As a result of these searches, we have found 2500 new cyclic codes and 730 negacyclic codes. These new codes exhibit improved parameters compared to previously known codes. Additionally, we have obtained binary quantum codes with good parameters from such $\mathbb{Z}_4$ codes.
\end{abstract}


\medskip


\section{Introduction and Motivation}

After the discovery that some good non-linear binary codes can be obtained as images of $\mathbb{Z}_4$-linear codes (also called quaternary codes in the literature) in \cite{Hamm1994}, research on codes over $\mathbb{Z}_4$ and codes over many other finite rings exploded. 
Codes over rings became a prominent topic in algebraic coding theory. Given the influence of the paper \cite{Hamm1994}, much subsequent research  on codes over  $\mathbb{Z}_4$, and the fact that many of the rings studied are extensions of  $\mathbb{Z}_4$, codes over $\mathbb{Z}_4$ attained a prominent place in algebraic coding theory. Therefore, a database of best known codes over $\mathbb{Z}_4$ was established \cite{aydin2009database}, similar to the one on linear codes over fields \cite{database}.  The initial database  \cite{aydin2009database} was relocated and expanded in a recent work \cite{database2022}. 

 In the process of updating the database, researchers found thousands of new linear codes over $\mathbb{Z}_4$ from cyclic codes of odd length, and quasi-cyclic (QC) codes that used those cyclic codes as building blocks (or seeds). This work is a continuation of the work done  in  \cite{database2022}. In this work, we aim to expand the work of our colleagues and find more  $\mathbb{Z}_4$-linear codes for the database. We conducted computer searches on cyclic codes of even length over $\mathbb{Z}_4$, both free  and non free,  as well as negacyclic codes of even length. As a result, we have found over 3000 linear codes over $\mathbb{Z}_4$ that are new to the database. By a new code,
we mean a code with better parameters than the previously
known code, or a construction of a code with those
parameters for the first time. We use the term ``new'' in the sense of code
parameters and previously documented constructions
rather than equivalence classes. We did not perform
a separate code-equivalence test. Additionally, we constructed binary quantum codes
derived from cyclic and negacyclic codes over
$\mathbb{Z}_4$ with best-known parameters. These quantum
codes are not new, but they have the same parameters as
the best-known quantum codes.

The paper is organized as follows. We start  with some basic facts and definitions in \cref{sec:bd} about codes over $\ZZ$. Then, we summarize basic facts about cyclic codes and some of their generalizations in section 3. We explain the theory we use in producing the algorithms in \cref{sec:sp}. Finally, we present the results of computer searches in \cref{sec:cr}.

\section{Basic Definitions}
\label{sec:bd}

\subsection{$\mathbb{Z}_4$ codes} 

A code  $\mathcal{C}$ of length $n$ over $\mathbb{Z}_4$ is a subset of \(\mathbb{Z}_4^n\), and a linear code  is a submodule (equivalently, an additive subgroup) of \(\mathbb{Z}_4^n\). Elements of $\mathcal{C}$ are called codewords. A matrix whose rows generate $\mathcal{C}$ is called a generator matrix.
Any linear code  over $\ZZ$ is equivalent to a code  with a generator matrix $G$ of the form
\begin{align}\label{eq:gen}
    \begin{bmatrix}
    I_{k_1} & A & B \\
    0 & 2I_{k_2} & 2C
    \end{bmatrix}
\end{align}
where $A$ and $C$ are $\mathbb{Z}_2$-matrices and $B$ is a $\ZZ$-matrix. We  say that $\mathcal{C}$ is of type $4^{k_1}2^{k_2}$, which is also the size of $\mathcal{C}$. We denote the parameters of $\mathcal{C}$ as $[n,k_1,k_2,d_L]$ or just $[n,k_1,k_2,d]$, where $n$ is the length of $\mathcal{C}$ and $d_L$ or $d$ is the minimum Lee distance of $C$. Besides the Lee distance, Euclidean and Hamming distances are also considered for codes over $\ZZ$, however, in this paper we primarily consider the Lee distance because it is the most relevant metric when we consider binary images of $\ZZ$-codes and compare them with existing binary codes. The Euclidean distance is important for the notion of Type II codes  \cite{type2}. A $\ZZ$-linear code is not
necessarily a free module, a module with a basis. It is free if and only if $k_2 = 0$. If $k_2$  is zero, we call $\mathcal{C}$ a free code. Otherwise,  $\mathcal{C}$ is non-free.

The Lee weights of \(0, 1, 2, 3 \in \mathbb{Z}_4\), denoted by \(w_L(0), w_L(1), w_L(2), w_L(3)\), respectively, are defined by
\[
w_L(0) = 0, \quad w_L(1) = w_L(3) = 1, \quad w_L(2) = 2.
\]

\noindent  The Lee weight \(w_L(x)\) of \(x = (x_1, x_2, \ldots, x_n) \in \mathbb{Z}_4^n\) is defined to be the integral sum of the Lee weights of its components:
\[
w_L(x) = \sum_{i=1}^n w_L(x_i).
\]

\noindent This weight function defines a distance function
\[
d_L(x, y) = w_L(x - y)
\]
on \(\mathbb{Z}_4^n\), which is called the Lee distance.

\noindent  The Gray map $\phi: \ZZ^n \rightarrow \mathbb{Z}_2^{2n}$ is the coordinate-wise extension of the bijection $\ZZ \rightarrow \mathbb{Z}_2^2$ defined by 
\begin{align*}
    0 \rightarrow 00, \\
    1 \rightarrow 01, \\
    2 \rightarrow 11, \\
    3 \rightarrow 10.
\end{align*}
The image $\phi(C)$ of a linear code $C$ over $\ZZ$ of length $n$ by the Gray map is a binary code of length $2n$. We know that $\phi$ is both  a weight-preserving map from
\begin{align*}
    (\ZZ^n, \text{Lee weight } (w_L)) \text{ to } (\mathbb{Z}_2^{2n}, \text{Hamming weight } (w_H))
\end{align*}
and a distance-preserving map from
\begin{align*}
    (\ZZ^n, \text{Lee distance } (d_L)) \text{ to } (\mathbb{Z}_2^{2n}, \text{Hamming distance } (d_H)).
\end{align*}
We also know that for any linear code $C$ over $\ZZ$, its Gray map image $\phi(C)$ is distance invariant, so its minimum Hamming distance $d_H$ is equal to the minimum Hamming weight $w_H$. Thus, we know that the minimum Hamming distance of $\phi(C)$ is equal to the minimum Lee weight $d_L$ of $C$. Additionally, we know that
\begin{align*}
    w_L = w_H = d_H = d_L.
\end{align*}
Thus, $\phi(C)$ is a binary code but not necessarily a linear code because the Gray map is not linear. Conditions that characterize binary codes that can be obtained as Gray images of $\ZZ$-linear codes are given in \cite{Wan1997} (Corollary 3.15 and Corollary 3.17).  We refer the reader to~\cite{Wan1997} for more details on codes over $\ZZ$.

One of the main problems of coding theory is to determine the optimal values of the parameters of a linear code and explicitly construct codes with those parameters. For codes over finite fields, there have been databases of best known linear codes (BKLC). For example, the online database \cite{database} is well known. The Magma software also contains such a database \cite{magmaOnline}. For codes over $\ZZ$, a database was introduced more recently \cite{database2022}.  Looking at the tables, we observe that in most cases optimal codes are not known. Researchers often use computer searches to find codes with better parameters than currently BKLCs. However, exhaustive searches over  all linear codes are not feasible due to two fundamental computational challenges.  First, computing the minimum distance of an arbitrary linear code is NP-hard \cite{NPHard}, and it quickly becomes infeasible for larger dimensions. Second, for a given length, dimension, and  the finite field $GF(q)$, the number of linear codes is too big to conduct exhaustive computer searches, except for small lengths and dimensions. Given these  computational complexity challenges, researchers often focus on promising subclasses of linear codes with rich mathematical structures. 

Our approach builds on existing theoretical descriptions of cyclic and negacyclic codes
over $\mathbb{Z}_4$. Rather than developing new structural characterizations of these code families, we use the established constructions to carry out computational searches for codes with improved or previously unrecorded parameters. In this respect, our work complements earlier computational studies that contributed additional codes to existing databases. Our focus is on cyclic and negacyclic codes of even length, their minimum Lee distances, and the binary quantum codes obtained from the resulting classical codes.

\section{Cyclic Codes and Their Generalizations}

Cyclic codes and their various generalizations play a key role in algebraic coding theory. Their structure makes them easier to implement in practice, and they establish a key link between algebra and coding theory.

\begin{definition}
A code $\mathcal{C}$ of length $ n $ over a finite field $ \mathbb{F}_q $ is called a cyclic code if for every codeword $ v = (v_0, v_1, \ldots, v_{n-1}) \in \mathcal{C}$, the word obtained by cyclically shifting $ v $ by one position to the right, that is, $ \sigma(v) = (v_{n-1}, v_0, v_1, \ldots, v_{n-2}) $, is also in $\mathcal{C}$. 
\end{definition}

 One useful generalization of cyclic codes is the class of quasi-cyclic (QC) codes. These are linear codes that are invariant under cyclic shifts by a fixed number of positions.

\begin{definition} A linear code $\mathcal{C}$ of length $n$ over a finite ring or field is called a quasi-cyclic (QC) code of index $\ell$ if for every codeword $ v = (v_0, v_1, \ldots, v_{n-1}) \in \mathcal{C}$, the word obtained by cyclically shifting $ v $ by $ \ell $ positions, that is, $ \sigma^{\ell}(v) = (v_{n-\ell}, v_{n-\ell+1}, \ldots, v_{n-1}, v_0, v_1, \ldots, v_{n-\ell-1}) $, is also in $\mathcal{C}$. 
\end{definition}

\begin{definition}
A linear code of length \(N\) over a ring $R$ is constacyclic if for some unit \(a \in R\), the code is invariant under the automorphism \(\sigma_a\) defined by

\[
(c_0, c_1, \ldots, c_{N-1}) \mapsto (a c_{N-1}, c_0,c_1, \ldots, c_{N-2}).
\]
\end{definition}

Note that the special case \(a = 1\) gives cyclic codes.  In the case \(a = -1\), we say that the code is \textit{negacyclic}. Constacyclic codes of length \(N\) over \(R\) can be identified as ideals in the quotient ring \(R[x]/\langle x^N - a\rangle\) via the isomorphism from \(R^N\) to \(R[x]/\langle x^N - a\rangle\) defined by

\[
(v_0, v_1, \ldots, v_{N-1}) \mapsto v_0 + v_1 x + \cdots + v_{N-1} x^{N-1}.
\]


The algebraic structure of cyclic and constacyclic codes over a finite field $\Fq$ is quite simple. They are ideals in the quotient ring  \(\Fq[x]/\langle x^N - a\rangle\) which is a principal ideal ring. For every constacyclic code $\C$ of length $N$ over $\Fq$, there is a unique monic polynomial $g(x)$ of least degree in  $\C$ that generates $\C$, i.e., $\C=\langle g(x)\rangle=\{f(x)g(x) \mod x^N-a: f(x)\in \Fq[x] \}$. This polynomial is called the standard generator polynomial of $\C$ and it divides $x^N-a$. Hence, there is a one-to-one correspondence between constacyclic codes of length $N$ over $\Fq$ and divisors of $x^N-a$ over $\Fq$.  When $\gcd(N,q)=1$, the polynomial $x^N-a$ has $N$ distinct roots (in an extension of $\Fq$), otherwise it has repeated roots. In either case, all constacyclic codes of length $N$ over $\Fq$ can be obtained from the divisors of $x^N-a$. 

\section{Cyclic Codes of Even Length  over $\mathbb{Z}_4$}
\label{sec:sp}

When we consider cyclic codes over $\ZZ$, their algebraic structure is more complicated. In the case $\gcd(N,4)=1$, that is when $N$ is odd, their structure is similar to the case of cyclic codes over $\Fq$. This is the easier case where the factorization of $x^N-1$ into irreducibles over $\ZZ$ is unique and can be obtained from the factorization of $x^N-1$ into irreducibles over $\mathbb{Z}_2$ via the Hensel lift \cite{Wan1997}. Let $N$ be the length of a $\ZZ$-linear code. It can uniquely be written as $N = 2^k \cdot n$ where $n$ is odd. We can classify the codes in $\mathbb{Z}_4$ based on the parity of code length $N$. If $k=0$ then $N = n $ is odd. This is the case considered in \cite{database2022} where the authors found a large number of new $\ZZ$-linear codes from cyclic codes of odd length over $\ZZ$. The structure of  cyclic codes of odd length over $\ZZ$ is given by the following theorem:
\begin{theorem}\cite{Aydin2011}
Let $C$ be a cyclic code in $R_n = \mathbb{Z}_4[x]/\langle x^n - 1\rangle$ where $n$ is odd. Then, $R_n$ is a principal ideal ring and

    \item 
    \[
    C =\langle (g(x), 2a(x))\rangle = \langle g(x) + 2a(x)\rangle
    \]
    where $g(x)$, $a(x)$ are polynomials with $a(x) \mid g(x) \mid (x^n - 1) \mod 4$.
    
\end{theorem}
 If $k \ge 1$ then $N$ is even. We consider two subcases :
    \begin{itemize}
        \item If $k = 1$, then $N$ is said to be oddly even, that is, $N$ is a multiple of 2 and not 4.
        \item If $k > 1$, then $N$ is a multiple of 4. 
    \end{itemize}
The first theorem we implemented only handled the oddly even case:
\subsection{Cyclic codes of oddly even length}
The case of cyclic codes of oddly even length is first studied in \cite{Blackford2003a}.
\begin{theorem}\cite{Blackford2003a}
Let  $N=2n$ where $n$ is odd, and $\mathcal{C}$ be an ideal in $\mathbb{Z}_4[x]/\langle x^{N} - 1\rangle$. Then $\mathcal{C}$ is generated by an element of the form
\[
\left( a_1(x^2) a_2(x^2) a_3(x^2) \tilde{b}(x) c(x), 2a_1(x^2) a_2(x) b(x) \right)
\]
where $x^n - 1 = a_1(x) a_2(x) a_3(x) b(x) c(x) d(x)$,  $a_1(x), a_2(x), a_3(x), b(x), c(x)$ and $d(x)$ are monic and pairwise relatively prime polynomials in $\mathbb{Z}_4[x]$, and $\tilde{b}(x)$ is a monic polynomial such that $\tilde{b}(x) \equiv b(x) \pmod{2\mathbb{Z}_4[x]}$. 
\end{theorem}
    The main challenge when handling codes with an even length $N$ over $\mathbb{Z}_4$ is that  $x^N - 1$ does not have a unique factorization into irreducibles and Hensel lifting does not work. Computer algebra systems such as Magma do not offer any implementation of factorization in this case.  Consequently, we must manually check every possible polynomial with a degree $\leq N$. The possibilities grows exponentially, and exhaustive searches  becomes intractable  for larger values of $N$ due to both CPU time and the memory need. 
    
    This theorem is useful as it allows us to structure our search for the cyclic codes of length $N=2\cdot n$ where  $n$ is odd. We can find the irreducible factors of $x^n - 1$ in $\mathbb{Z}_2$ and find their unique Hensel lifts  in $\mathbb{Z}_4$.
    
The factors are assigned to the component polynomials
$a_1(x)$, $a_2(x)$, $a_3(x)$, $b(x)$, $c(x)$, and $d(x)$.
For the exhaustive searches, we enumerate all possible
assignments of these factors. For each assignment, we
construct the corresponding generator polynomials and
generator matrix, compute the parameters $(k_1,k_2)$
and the minimum Lee distance using Magma, and retain
the best code found for each pair $(k_1,k_2)$.

For each polynomial $b(x)$, the possible lifts
$\tilde{b}(x)$ are obtained by varying its coefficients
over $\mathbb{Z}_4$ while preserving the congruence
$\tilde{b}(x)\equiv b(x)\pmod{2\mathbb{Z}_4[x]}$.
The leading coefficient remains $1$, so each lift
is monic. The search procedure generates and examines
all such coefficient choices for each selected $b(x)$,
rather than choosing the lifts randomly.

To illustrate the process, we use a specific example. Let $N = 6$ and $n = 3$. The cyclotomic factors (another term for irreducible factors) of $x^n - 1$ over $\mathbb{Z}_4$ are given by $x + 3$ and $x^2 + x + 1$. We assign $x + 3$ to $b(x)$ and $x^2 + x + 1$ to $d(x)$. The ideal $\mathcal{C}$ in $\mathbb{Z}_4[x] / \langle x^6 - 1\rangle$ is then formed by structuring the remaining factors $a_1(x)$, $a_2(x)$, $a_3(x)$, and $c(x)$ to the identity element of $\mathbb{Z}_4$ as there is not enough distinct factors. In other words, we set $a_1(x) = a_2(x) = a_3(x) = c(x) = 1$.

\noindent Given:
\begin{align*}
a_1(x) &= 1 \\
a_2(x) &= 1 \\
a_3(x) &= 1 \\
b(x) &= x + 3 \\
c(x) &= 1 \\
d(x) &= x^2 + x + 1
\end{align*}

\noindent The ideal $\mathcal{C}$ is of the form
\[
\mathcal{C} = \langle a_1(x^2) a_2(x^2) a_3(x^2) \tilde{b}(x) c(x), 2a_1(x^2) a_2(x) b(x) \rangle.
\]

\noindent Substituting the values
\[
\mathcal{C} = \langle 1 \cdot 1 \cdot 1 \cdot \tilde{b}(x) \cdot 1, 2 \cdot 1 \cdot (x + 3) \rangle = \langle \tilde{b}(x), 2(x + 3) \rangle.
\]

 \noindent Here, $\tilde{b}(x)$ is a monic polynomial satisfying $\tilde{b}(x)\equiv b(x) \pmod{2\mathbb{Z}_4[x]}$.  Since $b(x) = x + 3$, we have $\tilde{b}(x) = x + 1$ or $\tilde{b}(x) = x+3$ which we will choose for the sake of simplicity.

\noindent So, the ideal $\mathcal{C}$ becomes
\[
\mathcal{C} = \langle x + 3, 2(x + 3) \rangle
\]

\noindent Finally, the polynomials that generate an ideal $\mathcal{C}$ for $N = 6$ and $n = 3$ are $x + 3$ and $2x + 2$. Therefore, the generating polynomials are:

\begin{align*}
f_1(x) &= x + 3, \\
f_2(x) &= 2(x + 3)
       = 2x + 6
       \equiv 2x + 2 \pmod{4}.
\end{align*}

\noindent The resulting cyclic code is a code with $k_1=5$ and $k_2=0$ and minimum Lee distance of 2. This code differs from the length-six code listed in Table 3. In this example, the generator polynomials are $f(x)=x+3$ and $g(x)=2(x+3)$, whereas the code in Table 3 has $f(x)=x+3$ and $g(x)=2$. Our Magma computations confirm that these codes contain 1024 and 2048 codewords, respectively, showing that they are distinct codes. The parameters of this code, and all of the codes obtained in this work have been computed using the Magma software.

The computations were performed on a Linux server
running Ubuntu 24.04.4 LTS. The server has 64 CPUs
using AMD Opteron 6378 processors and approximately
251 GiB of RAM. We used Magma V2.29-7 for polynomial
factorization, generator construction, and code
parameter computations.

The minimum Lee distances reported in this paper
were computed using Magma's built-in code computations.
We did not implement a separate minimum-distance
algorithm or examine the internal implementation of
Magma's routines. The reported code parameters can
be independently verified by constructing the codes
from the generator polynomials provided.

This structured approach helps in reducing the computational complexity and aids in identifying the valid cyclic codes over $\mathbb{Z}_4$. However, as $N$ grows larger and as shown in the table below, the number of ideals grows exponentially. Thus, exhaustive searches are only possible for lower lengths.

For oddly even cyclic codes, we carried out
exhaustive searches up to length 48. For larger
oddly even lengths up to 98, exhaustive searches
became computationally impractical. We therefore
sampled 10,000 candidate codes for each length
instead of enumerating the entire search space.

For the random searches, we used Magma's default
randomization without setting a fixed seed in our
scripts. Magma records the seed used for each
session in its output. We did not perform repeated
runs specifically to assess the stability of
the random searches.

With the oddly even numbers out of the way, we shifted our attention to the case where $N=2^k \cdot n$ such that $k > 1$. Since we already found the codes such that $k = 1$, we restricted our search for the general case to $k > 1$. The total number of codes at each length was determined using a formula from \cite{Blackford2003a}, while the numbers of free codes and codes generated by a single polynomial were determined through exhaustive computational searches. These searches became computationally infeasible for larger lengths.  This is why many entries for even lengths are blank in Table 1.  In Table 1, ``\#Cyclic'' denotes the total number of cyclic codes of length $N$ over $\mathbb{Z}_4$, ``\#Free'' denotes the number of free cyclic codes, and ``\#1-Gen'' denotes the number of cyclic codes generated by a single polynomial. Blank entries indicate counts that were not determined by our computational searches.

\begin{table}[ht]
\caption{The number of cyclic codes, free cyclic codes, and cyclic codes with one generator polynomial over $\mathbb{Z}_4$ for $1\leq N\leq 100$}
\centering
\begin{tabular}{|l|c|c|r|}
\hline
$N$ & $\#$Cyclic & $\#$Free & $\#$1-Gen \\
\hline
1  & 3         & 2      & 3     \\
\hline
2  & 7         & 4      & 6     \\
\hline
3  & 9         & 4      & 9     \\
\hline
4  & 23        & 10     & 16    \\
\hline
5  & 9         & 4      & 9     \\
\hline
6  & 63        & 24     & 48    \\
\hline
7  & 27        & 8      & 27    \\
\hline
8  & 135       & 46     & 76    \\
\hline
9  & 27        & 8      & 27    \\
\hline
10 & 147       & 72     & 120   \\
\hline
11 & 9         & 4      & 9     \\
\hline
12 & 1035      & 260    & 576   \\
\hline
13 & 9         & 4      & 9    \\
\hline
14 & 1183      & 400    & 857   \\
\hline
15 & 243       & 32     & 243   \\
\hline
16 & 2519      &        &       \\
\hline
17 & 27        & 8      & 27    \\
\hline
18 & 4347      &        &       \\
\hline
19 & 9         & 4      & 9     \\
\hline
20 & 7935      &        &       \\
\hline
21 & 729       & 64     & 729   \\
\hline
22 & 7203      &        &       \\
\hline
23 & 27        & 8      & 27    \\
\hline
24 & 106515    &        &       \\
\hline
25 & 27        & 8      & 27    \\
\hline
26 & 28707     &        &       \\
\hline
27 & 81        & 16     & 81    \\
\hline
28 & 293687    &        &       \\
\hline
29 & 29        & 4      & 29    \\
\hline
30 & 583443    &        &       \\
\hline
31 & 2187      & 128    & 2187  \\
\hline
32 & 655287    &        &       \\
\hline
33 & 243       & 32     & 243   \\
\hline
34 & 476847    &        &       \\
\hline
35 & 729       & 64     & 729   \\
\hline
36 & 4579875   &        &       \\
\hline
37 & 9         & 4      & 9     \\
\hline
38 & 1835043   &        &       \\
\hline
39 & 243       & 32     & 243   \\
\hline
40 & 11953575  &        &       \\
\hline
41 & 27        & 8      & 27    \\
\hline
42 & 50690367  &        &       \\
\hline
43 & 81        & 16     & 81    \\
\hline
44 & 24235215  &        &       \\
\hline
45 & 6561      & 256    & 6561  \\
\hline
46 & 29503663  &        &       \\
\hline
47 & 27        & 8      & 27    \\
\hline
48 & 513566163 &        &       \\
\hline
49 & 243       & 32     & 243   \\
\hline
50 & 154141407 &        &       \\
\hline
\end{tabular}
\hspace{1cm}
\begin{tabular}{|l|c|c|r|}
\hline
$n$ & $\#$Cyclic & $\#$Free & $\#$1-Gen \\
\hline
51 & 6561      & 256    & 6561  \\
\hline
52 & 386347215 &        &       \\
\hline
53 & 9         & 4      & 9     \\
\hline
54 & 1139561703&        &       \\
\hline
55 & 243       & 32     & 243   \\
\hline
56 & 7298992215&        &       \\
\hline
57 & 243       & 32     & 243   \\
\hline
58 & 1879048227&        &       \\
\hline
59 & 9         & 4      & 9     \\
\hline
60 & 42500851875&       &       \\
\hline
61 & 9         & 4      & 9     \\
\hline
62 & 17960084863&       &       \\
\hline
63 & 1594323   & 8192   & 1594323 \\
\hline
64 & 42949672823  &       &       \\
\hline
65 & 2187 &   128    &   2187    \\
\hline
66 & 68641485507&       &       \\
\hline
67 & 9 &   4    &   9    \\
\hline
68 & 1.03E+11 &       &       \\
\hline
69 & 729 &   64    &  729     \\
\hline
70 & 4.18E+11&       &       \\
\hline
71 & 27 &   8    &   27    \\
\hline
72 & 1.93E+12&       &       \\
\hline
73 & 19683&   512    &  19683     \\
\hline
74 & 4.81E+113&       &       \\
\hline
75 & 6561&   256    &  6561     \\
\hline
76 & 1.58E+12&     &      \\
\hline
77 & 729 &   64    &  729   \\
\hline
78 & 4.35E+12&    &      \\
\hline
79 & 27 &   8    &  27    \\
\hline
80 & 1.46E+13 &     &      \\
\hline
81 & 243 &  32   &  243    \\
\hline
82 & 7.70E+12 &     &      \\
\hline
83 & 9 &   4  &   9   \\
\hline
84 & 2.59E+14 &     &     \\
\hline
85 & 531441 &  4096  &  531441 \\
\hline
86 & 3.08E+13 &  &   \\
\hline
87 & 243 & 32 & 243  \\
\hline
88 & 1.49E+14 &  &   \\
\hline
89 & 19683 & 512 & 19683 \\
\hline
90 & 6.77E+14 &  &   \\
\hline
91 & 59049 & 1024 & 59049  \\
\hline
92 & 4.07E+14 &  &   \\
\hline
93 & 4782969 & 16384 & 4782969  \\
\hline
94 & 4.93E+14 &  &   \\
\hline
95 & 243 & 32 &  243 \\
\hline
96 & 8.76E+15 &  &  \\
\hline
97 & 27 & 8 &  27 \\
\hline
98 & 5.20E+15 &  &   \\
\hline
99 & 6561 & 256 &  6561 \\
\hline
100 & 8.72E+15 &  &  \\
\hline

\end{tabular}
\end{table} 

\clearpage

\subsection{General case of cyclic codes of even length}
The form of a generator of a cyclic code of even length in the general  case is given in a theorem in \cite{Abualrub2007}. 

\begin{theorem}\cite{Abualrub2007}
\begin{enumerate} 
    \item Let $\C$ be a cyclic code of even length $N$ over $\ZZ$. Then either
    \begin{enumerate}
        \item $\C$ is a free module with a single generator
        \[
        \C = \langle g(x) + 2p(x)\rangle,
        \]
        where $g(x) \mid (x^n - 1) \mod 2$ and $(g(x) + 2p(x)) \mid (x^N - 1) \mod 4$, or,
        \item $\C = \langle (g(x) + 2p(x), 2a(x))\rangle$ where $g(x)$, $a(x)$, and $p(x)$ are polynomials with $g(x) \mid (x^N - 1) \mod 2$, $a(x) \mid g(x) \mod 2$, $a(x) \mid p(x) \left(\frac{x^N - 1}{g(x)}\right) \mod 2$, and $\deg( a(x))>\deg(p(x))$.
    \end{enumerate}
\end{enumerate}
\end{theorem}

To generate all free cyclic codes for a given \( N \) over \(\mathbb{Z}_4\), we need to exhaustively combine the divisors \( g(x) \) of \( x^N - 1 \) over \(\mathbb{Z}_2\). For each \( g(x) \), we then explore the polynomial vector space to find a polynomial \( p(x) \) with coefficients in \(\mathbb{Z}_4\) and of degree less than $N$. The goal is to find \( p(x) \) such that \( g(x) + 2p(x) \mid (x^N - 1) \mod 4 \).

For each valid generator polynomial $g(x)+2p(x)$,
we construct the corresponding cyclic code and use Magma
to compute its parameters and minimum Lee distance.
For non-free cyclic codes, we similarly construct the
generators $g(x)+2p(x)$ and $2a(x)$ subject to the
conditions in Theorem 4.3, and compute the parameters
and minimum Lee distance of each resulting code.

For example, consider \( N = 6 \) and \( g(x) = x^2 + x + 1 \), which is a divisor of \( x^6 - 1 \) in \(\mathbb{Z}_2\). Through exhaustive search, we find that a possible \( p(x) \) is \( x \). Thus, \( g(x) + 2p(x) = x^2 + x + 1 + 2x = x^2 + 3x + 1 \). We verify that \( x^2 + 3x + 1 \) divides \( x^6 - 1 \mod 4 \). In fact,

\[
x^6 - 1 \equiv (x^2 + 3x + 1) \cdot Q(x) \mod 4 \text{ and}
\]
\[
Q(x) \equiv x^4 + x^3 + 3x + 3
\]

Next, we create a circulant matrix of length 6 from the polynomial \( x^2 + 3x + 1 \):

\[
\begin{pmatrix}
1 & 3 & 1 & 0 & 0 & 0 \\
0 & 1 & 3 & 1 & 0 & 0 \\
0 & 0 & 1 & 3 & 1 & 0 \\
0 & 0 & 0 & 1 & 3 & 1 \\
1 & 0 & 0 & 0 & 1 & 3 \\
3 & 1 & 0 & 0 & 0 & 1 \\
\end{pmatrix}
\]

\noindent The cyclic code generated by this matrix has parameters \( k_1 = 4 \), \( k_2 = 0 \), and a minimum Lee weight of 2.

\subsection{Negacyclic codes}
Considering constacyclic codes over $\ZZ$ that are not cyclic codes, there is only one unit in $\ZZ$ different from 1 which is $a=3=-1$. Hence, such codes are negacyclic codes. Moreover, it is well known that  negacyclic codes of odd length are equivalent to cyclic codes of odd length over $\ZZ$ via the isomorphism $f(x)\to f(-x)$ \cite{aydin2002}. Therefore, it suffices to consider negacyclic codes of even length over $\ZZ$. Their structure is given in the following theorem.
\begin{theorem} \cite{Blackford2003b}
Let \(\C\) be a negacyclic code over \(\mathbb{Z}_4\) of length \(N = 2^a n\), \(n\) odd. If \(\C\) is considered as an ideal in \(\mathbb{Z}_4[x]/\langle x^N + 1\rangle\), then \(\C = \langle g(x) \rangle\), where
\[
g(x) = \prod_{i=0}^{2^a+1} \left[ g_i(x) \right]^i
\]
and \(\{g_i(x)\}\) are monic co-prime divisors of \(x^n - 1\) in \(\mathbb{Z}_4[x]\). 

\end{theorem}
This is another theorem that helps significantly with computational efficiency. With this theorem, we can find all negacyclic codes for a particular length \( N \) by reducing the search to an odd \( n = \frac{N}{2^a} < N \), saving computational power and allowing easy access to the divisors through Hensel lifting. To illustrate the process, we  examine \( N = 6 \), which means \( a = 1 \) and \( n = 3 \).

The two factors for \( x^3 - 1 \) are \( x + 3 \) and \( x^2 + x + 1 \). We choose \( g_0(x) = x + 3 \) and \( g_1(x) = x^2 + x + 1 \). Since in this case, we have two divisors, which is less than the required number of terms \( 2^a + 2 = 2^1 + 2 = 4 \), we will take \( g_2(x) = 1 \) and \( g_3(x) = 1 \). For our example, the terms  are
\[
g_0(x) = x + 3
\]
\[
g_1(x) = x^2 + x + 1
\]
\[
g_2(x) = 1
\]
\[
g_3(x) = 1
\]

\noindent Now, combining these terms, we obtain 
\[
g(x) = (g_0(x))^0 (g_1(x))^1 (g_2(x))^2 (g_3(x))^3 = (x + 3)^0 (x^2 + x + 1)^1 (1)^2 (1)^3 = x^2 + x + 1.
\]

Next, we create the nega-circulant matrix of size 6 from this polynomial. Each successive row is obtained by shifting the previous row one position to the right and multiplying the wrapped entry by $-1$ modulo $4$.
\[
\begin{pmatrix}
1 & 1 & 1 & 0 & 0 & 0 \\
0 & 1 & 1 & 1 & 0 & 0 \\
0 & 0 & 1 & 1 & 1 & 0 \\
0 & 0 & 0 & 1 & 1 & 1 \\
3 & 0 & 0 & 0 & 1 & 1 \\
3 & 3 & 0 & 0 & 0 & 1 \\
\end{pmatrix}
\]

\noindent The negacyclic code generated by this matrix has  the parameters $
k_1 = 4, k_2 = 2$ and  the minimum Lee weight  2.

For each even length up to 30, we enumerate the
possible assignments of the cyclotomic factors to
the polynomials $g_i(x)$ in Theorem 4.4. For each
assignment, we construct the corresponding generator
polynomial and use Magma to compute the code parameters
and minimum Lee distance. This search is exhaustive
over the specified factor assignments. After exhaustively obtaining all negacyclic codes for
all even lengths up to 30, we investigated various
methods of obtaining quantum codes from the
$\mathbb{Z}_4$ codes we constructed.

\subsection{Quantum codes} 
Building large scale quantum computers  is a very active area of research. A critical part of this effort is the ability to control quantum errors. There are various methods of constructing quantum error correcting codes (QECC) over finite fields $\Fq$, many of them based on the CSS construction \cite{QECC4}.  Codes over extension rings of $\Fq$ have been utilized to construct QECCs over $\Fq$. A large number of works are available in the literature on this approach.  Recently, some binary QECCs have been obtained from the extension rings of $\ZZ$ in  \cite{Dinh2022} and \cite{quantumZ4Ext}.  One of our goals in this work was to obtain some binary quantum codes from codes over $\ZZ$, rather than  from codes over an extension ring of $\ZZ$ as done in the recent papers. We used a method  similar \cite{Dinh2022} and \cite{quantumZ4Ext}. Additionally, we examined the Gray image of a \(\mathbb{Z}_4\)-linear code and applied CSS construction to the image as well as the binary codes associated to any $\ZZ$-linear codes:  torsion and residue codes.

To explain CSS construction, we recall some relevant basic facts about classical codes. For any \( u = (u_1, \ldots, u_n) \) and \( v = (v_1, \ldots, v_n) \) in \( \mathbb{F}_q^n \), the (Euclidean) inner product \( u \cdot v \) is defined as
\[ u \cdot v = u_1 v_1 + u_2 v_2 + \cdots + u_n v_n. \]
The dual code of \( C \) is defined as
\[ C^\perp = \{ v \in \mathbb{F}_q^n : u \cdot v = 0 \text{ for all } u \in C \}. \]
A code $C$ that is contained in its dual, \( C \subseteq C^\perp \),  is called self-orthogonal or weakly self-dual. A code $C$ that contains its dual, \( C^\perp  \subseteq C\), is called dual-containing. If \( C = C^\perp \) then we say that \( C \) is self-dual. If $C \cap C^\perp =\{ 0\}$ then $C$ (and $C^\perp$) is called  a linear complementary dual (LCD) code.

\subsection{CSS Construction}
We can apply the  CSS construction to obtain QECCs  from cyclic, quasi-cyclic, and negacyclic codes. For this, we need two codes such that one is contained in the dual of the other one. Hence for cyclic and constacyclic codes,  this condition is characterized by the ideal inclusion
\begin{align*}
    \langle g(x) \rangle \supseteq \langle g(x)f(x) \rangle.
\end{align*}
So we take $C_2^{\perp}=\langle g(x)f(x)\rangle \subseteq \langle g(x) \rangle =C_1 $. 

\begin{theorem}[CSS construction \cite{QECC4}]
Let  $C_1$ and $C_2$ be two linear codes over $\mathbb{F}_q$ with parameters $[n,k_1,d_1]_q$ and $[n,k_2,d_2]_q$ with $C_2^{\perp}\subseteq C_1.$ Then there exists a QECC with parameters $[[n, k_1 + k_2 - n, \min(d_1, d_2)]]_q$. In case $C_1$ is a dual-containing code, that is, $C_1^{\perp} \subseteq C_1$, there exists a QECC with parameters $[[n, 2k_1 - n, d_1]]_q$.
 \end{theorem}

\vspace{2em}

\subsection{Binary Residue and Torsion Codes of a $\mathbb{Z}_4$ Code} 

For every linear code over $\ZZ$, there are two binary linear codes associated with it and they posses a property that makes them suitable in applying the CSS construction. 
In order to identify these codes, \cite{Wan1997} introduced the following three maps \(\alpha\), \(\beta\), \(\gamma\) from \(\mathbb{Z}_4\) to \(\mathbb{Z}_2\).

\begin{table}[h!]
\centering
\begin{tabular}{|c|c|c|c|}
\hline
\(\mathbb{Z}_4\) & \(\alpha\) & \(\beta\) & \(\gamma\) \\ \hline
0 & 0 & 0 & 0 \\ \hline
1 & 1 & 0 & 1 \\ \hline
2 & 0 & 1 & 1 \\ \hline
3 & 1 & 1 & 0 \\ \hline
\end{tabular}
\caption{}
\end{table}

\noindent Clearly, \(\alpha\) is an additive group homomorphism from \(\mathbb{Z}_4\) to \(\mathbb{Z}_2\), but \(\beta\) and \(\gamma\) are not. Each element \(x \in \mathbb{Z}_4\) has a 2-adic expansion
\[
x = \alpha(x) + 2\beta(x).
\]
We also have
\[
\alpha(x) + \beta(x) + \gamma(x) = 0 \mod 2 \quad \text{for all } x \in \mathbb{Z}_4.
\]

\noindent The maps \(\alpha\), \(\beta\), \(\gamma\) can be extended to \(\mathbb{Z}_4^n\) in an obvious way. For \(x = (x_1, \ldots, x_n) \in \mathbb{Z}_4^n\), define
\[
\alpha(x) = (\alpha(x_1), \ldots, \alpha(x_n)),
\]
\[
\beta(x) = (\beta(x_1), \ldots, \beta(x_n)),
\]
\[
\gamma(x) = (\gamma(x_1), \ldots, \gamma(x_n)).
\]
Then \(\phi\) is extended to \(\mathbb{Z}_4^n\) as follows:
\[
\phi(x) = (\beta(x), \gamma(x)) \quad \text{for all } x \in \mathbb{Z}_4^n. 
\]
Clearly, the extended \(\phi\) is a bijection from \(\mathbb{Z}_4^n\) to \(\mathbb{Z}_2^{2n}\). For any \(x \in \mathbb{Z}_4^n\), \(\phi(x)\) is called the \textit{binary image} of \(x\) under \(\phi\).

\begin{proposition}
\cite[Proposition 3.18]{Wan1997} 
Let \( C \) be a \( \mathbb{Z}_4 \)-linear code of length \( n \), and \( C^{(1)} \) and \( C^{(2)} \) be the binary codes defined by:
\[
C^{(1)} = \{ \alpha(c) \mid c \in C \}
\]
and
\[
C^{(2)} = \{ \beta(c) \mid c \in C, \alpha(c) = 0 \}
\]
\begin{enumerate}
    \item[(i)] Both \( C^{(1)} \) and \( C^{(2)} \) are binary linear codes, and \( C^{(1)} \subseteq C^{(2)} \).
    \item[(ii)] If \( C \) is of type \( 4^{k_1}2^{k_2} \) and has generator matrix (\ref{eq:gen}), then the residue code \( C^{(1)} \) is a binary linear \([n, k_1]\)-code with generator matrix
    \[
    \left( I_{k_1} \quad A \quad \alpha(B) \right) 
    \]
    and the torsion code \( C^{(2)} \) is a binary linear \([n, k_1 + k_2]\)-code with generator matrix
    \[
    \left( \begin{array}{ccc}
    I_{k_1} & A & \alpha(B) \\
    & I_{k_2} & C
    \end{array} \right). 
    \]
\end{enumerate}
\end{proposition}

\noindent Since \(C^{(1)} \subseteq C^{(2)}\), they can be directly used in the CSS construction. There is also a specific CSS construction that could be used with this condition:

\begin{theorem} \cite[Theorem 9]{QECC4}
{
    Let $C_1 \subseteq C_2$ be binary linear codes. We obtain an $\left[\left[n, k_2 - k_1, d\right]\right]$ code, where
\[
d = \min\{\text{dist}(C_2 \setminus C_1), \text{dist}(C_1^\perp \setminus C_2^\perp)\}.
\]
}
\end{theorem}

\section{Computational Results}
\label{sec:cr}

Our computer searches yielded many new linear codes over \(\mathbb{Z}_4\). A code  \(\mathbb{Z}_4\)-linear code is new if either its parameters do not appear in the database \cite{Z4database}, or are better than the parameters of the comparable code on the database. Of the new codes we have found,  2500 of them are  from cyclic codes of  oddly even lengths. Table 2 below shows a small sample of these new codes. In the tables below, a polynomial is represented as a list of coefficients in descending order of the terms
to save space. For example, the string ``13'' on the first row of Table 3 represents the polynomial \( f(x) = x+3 \).

In Table 3, the entries $f$ and $g$ represent the complete
generator polynomials of the corresponding codes, as described
in Theorem 4.2, rather than the individual polynomial factors.
Specifically, these generators are of the form
$f(x)=a_1(x^2)a_2(x^2)a_3(x^2)\tilde{b}(x)c(x)$
and $g(x)=2a_1(x^2)a_2(x)b(x)$.
The polynomials are written as strings of coefficients in
descending order to save space. For example, the entry
``13'' represents the polynomial $x+3$.
 A total of 2500 record breaking codes have been identified in the category of oddly even length codes.
A link to the complete list of codes obtained in our computational searches, including their parameters
and generator polynomials, is provided as machine-readable supplementary material in Section 5.2. The tables in this section
present selected examples from these results.
As for the free codes in Table 4, each code's generator $f$ is chosen such that $f(x)=g(x) + 2p(x)$ is a divisor of $x^N-1$ which is the first case of Theorem 4.3. In this case, we have been able to obtain codes with best known parameters. No record-breakers were found mainly because most free codes up until our exhaustive search cap of 38  were already part of the database.
Finally, and for the negacyclic codes in Table 5, each code can be generated from the polynomial $g$, constructed according to  Theorem 4.4. A total of 730 record breakers have been identified in the negacyclic category.
\begin{table}[ht]
\caption{A sample of new codes from cyclic codes of oddly even length}
\centering
\renewcommand{\arraystretch}{1.5}
\setlength{\tabcolsep}{15pt} 
\begin{tabular}{|l|c|c|}
\hline
$[N,k_1,k_2,d]$  & $f$ & $g$  \\
\hline
$[6, 5, 1, 2]$ & $13$ & $2$ \\
\hline
$[10, 6, 4, 2]$ & $11111$ & $2$  \\
\hline
$[14, 4, 1, 10]$ & $10310221033$ & $2202222002$ \\
\hline
$[18, 11, 1, 4]$ & $11213013$ & $2002002$  \\
\hline
$[22, 10, 10, 4]$ & $1100000000033$ & $202$ \\
\hline
$[26, 12, 12, 4]$ & $110000000000033$ & $202$  \\
\hline
$[30, 14, 4, 8]$ & $11200000022020213$ & $2020020200202$ \\
\hline
$[34, 16, 10, 4]$ & $1302000000202220213$ & $222020222$ \\
\hline
$[38, 19, 1, 4]$ & $12200000000000000003$ & $2222222222222222222$ \\
\hline
\end{tabular}
\end{table}

\begin{table}[ht]
\caption{ Some good codes from  free cyclic codes of even lengths}
\centering
\renewcommand{\arraystretch}{1.5}
\setlength{\tabcolsep}{15pt} 
\begin{tabular}{|l|c|}
\hline
$[n,k_1,k_2,d]$ & $f$  \\
\hline
$[4, 1, 0, 4]$  & $3131$ \\
\hline
$[6, 2, 0, 6]$  & $13211$ \\
\hline
$[8, 1, 0, 8]$  & $31313131$ \\
\hline
$[10, 4, 0, 8]$  & $3320211$ \\
\hline
$[12, 2, 0, 8]$  & $33011033011$ \\
\hline
$[14, 3, 0, 12]$  & $311230013321$ \\
\hline
$[16, 2, 0, 16]$  & $121012101210121$ \\
\hline
$[18, 2, 0, 18]$  & $13211013211013211$ \\
\hline
$[20, 8, 0, 8]$  & $1212202202101$ \\
\hline
$[22, 10, 0, 8]$  & $3100202200011$ \\
\hline
\end{tabular}
\end{table}

\begin{table}[ht]
\caption{Record breaking negacyclic codes}
\centering
\renewcommand{\arraystretch}{1.5}
\setlength{\tabcolsep}{15pt} 
\begin{tabular}{|l|c|}
\hline
$[n,k_1,k_2,d]$ & $g$  \\
\hline
$[2, 0, 1, 4]$  & $22$ \\
\hline
$[4, 0, 1, 8]$  & $2222$ \\
\hline
$[6, 2, 2, 4]$  & $132310$ \\
\hline
$[10, 4, 4, 4]$  & $113303311$ \\
\hline
$[12, 1, 3, 6]$  & $331111333311$ \\
\hline
$[14, 0, 9, 8]$  & $2220020000000$ \\
\hline
$[18, 1, 7, 8]$  & $131313133111111113$ \\
\hline
$[20, 2, 6, 8]$  & $12101210303201032101$ \\
\hline
$[22, 10, 10, 4]$  & $113300000003311$ \\
\hline
$[34, 0, 10, 20]$  & $2222222220002022200000000022202000$ \\
\hline
$[36, 2, 8, 16]$  & $303230323030321210301230123010323210$ \\
\hline
$[38, 1, 1, 38]$  & $31133113311331133113311331133113311331$ \\
\hline
\end{tabular}
\end{table}

\newgeometry{left=1cm, right=1cm, top=2cm, bottom=2cm}
 \subsection{Quantum Codes}  In addition to the classical $\mathbb{Z}_4$-linear codes presented above, we construct binary quantum error-correcting codes from the classical codes obtained in our computational searches. These are two distinct contributions: the classical search produces the $\mathbb{Z}_4$-linear codes, while the quantum constructions use their associated binary residue and torsion codes, denoted by $C^{(1)}$ and $C^{(2)}$, respectively, together with the CSS construction described in Section 4.6.  We also obtained quantum codes from $\ZZ$ cyclic, quasi-cyclic, and negacyclic codes with even length. 
By applying the CSS construction to classical
codes over $\mathbb{Z}_4$, we obtained quantum
error-correcting codes whose parameters attain
the quantum Singleton bound in certain cases
or match the best-known parameters reported
in the literature \cite{database}. None of
these quantum codes improves the previously
known best parameters. We have not tested
whether the constructed quantum codes are
equivalent to existing codes, and therefore
we do not claim that they represent new
quantum codes up to equivalence.
 These codes are presented in Table 6.
The column ``Binary Code Parameters'' gives
the parameters $[n,k,d]$ of the classical
binary code used in the CSS construction.
The properties in the last column refer to
the classical binary codes whose parameters
are listed in the fourth column rather than
to the Gray images of the original
quaternary codes. Here, R denotes reversible,
DC denotes dual-containing, and SO denotes
self-orthogonal.

Throughout this subsection, ``quantum MDS''
refers to codes attaining the quantum Singleton
bound, while ``best-known'' refers to codes
whose parameters match the best values
documented in the cited literature. These
descriptions do not imply that the quantum
codes are new up to equivalence.

\renewcommand{\arraystretch}{2.0}
\begin{table}[h]
\begin{center}
\begin{minipage}{\textwidth}
\scriptsize
\setlength{\tabcolsep}{2pt}
\caption{Quantum Codes from classical Z4 codes with even length }\label{tab3}
\begin{tabular*}{\textwidth}{@{\extracolsep{\fill}}|l|l|l|l|l|l|}
\toprule
$[[n,k,d]]$ & $f$ & $g$ & Binary Code Parameters & References & Properties of Gray Images \\
\hline
$[[12,8, 2]]^{+}$ & 131313 & 222222&[12, 2, 6]&\cite{database} & R and SO\\
\hline
$[[12,10, 2]]^{*+}$ & 13 & 2 &[12, 11, 2]&\cite{database} & DC and SO\\
\hline
$[[20,18, 2]]^{*+}$ &0&2222222222&[20, 1, 20]&\cite{database} & R and SO\\
\hline
$[[20,16, 2]]^{+}$ &1111111111&2222222222&[20, 2, 10]&\cite{database}& R and SO\\
\hline
$[[20,2, 6]]^{*}$ &1120233&202&[20, 9, 4]&\cite{database} & R and SO \\
\hline

$[[28,22, 2]]$ &1213&2&[28, 25, 2]&\cite{database}&DC  \\
\hline
$[[28,24, 2]]$ &11111111111111&22222222222222&[28, 2, 14]&\cite{database}& R and SO \\
\hline
$[[28,26, 2]]^{*}$ &13&2&[28, 27, 2]&\cite{database} & R and DC \\
\hline
$[[36,30, 2]]$ &0&2002002002002002&[36, 3, 12]&\cite{database} & R and SO \\
\hline
$[[36,32, 2]]$  &111&2&[36, 34, 2]&\cite{database}&R and DC \\
\hline
$[[36,34, 2]]^{*}$ &13&2&[36, 35, 2]&\cite{database} &R and DC \\
\hline

$[[44,40, 2]]^{+}$ &1313131313131313131313&2222222222222222222222&[44, 2, 22]&\cite{database} &R and SO \\
\hline
$[[44,42, 2]]^{*+}$ &13&2&[44, 43, 2]&\cite{database} &R and DC\\
\hline
$[[52,48, 2]]$ &0&2020202020202020202020202&[52, 2, 26]&\cite{database} & R \\
\hline
\end{tabular*}
\end{minipage}
\end{center}
*: best known quantum codes.\\
+: these codes are also obtainable using negacyclic codes.
\end{table}


\clearpage

\newgeometry{left=1cm, right=1cm, top=2cm, bottom=2cm}

\noindent We also obtained quantum codes with the same parameters as best-known codes using Proposition 3.18 from \cite{Wan1997}. These are given in Table 7.
\renewcommand{\arraystretch}{2}
\begin{table}[h]
\begin{center}
\begin{minipage}{\textwidth}
\caption{Quantum Codes from the proposition 3.18}\label{tab3}
\begin{tabular*}{\textwidth}{@{\extracolsep{\fill}}| l|l|l|l|l|l|@{\extracolsep{\fill}}}
\toprule
$[[n,k,d]]$  & $f$ & $g$ & $C^{(1)}$ & $C^{(2)}$ & $C$\\
\hline
$[[6, 4, 2]]$ & 131313 & 22 &[6, 1, 6]& [6, 5, 2] &$((6, 4^1$ $  2^{4}))$\\
\hline
$[[10, 8, 2]]$ & 1111111111 & 22 &[10, 1, 10]& [10, 9, 2] &$((10, 4^1$ $  2^{8}))$ \\
\hline
$[[14, 12, 2]]$ & 11111111111111 & 22 &[14, 1, 14]& [14, 13, 2] &$((14, 4^1$ $  2^{12}))$ \\
\hline
$[[18, 16, 2]]$ & 131313131313131313 & 22 &[18, 1, 18]& [18, 17, 2] &$((18, 4^1$ $  2^{16}))$ \\
\hline
$[[22, 20, 2]]$ & 1313131313131313131313 & 22 &[22, 1, 22]& [22, 21, 2] &$((22, 4^{1}$ $  2^{20}))$ \\
\hline
$[[26, 24, 2]]$ & 13131313131313131313131313 & 22 &[26, 1, 26]& [26, 25, 2] &$((26, 4^{1}$ $  2^{25}))$ \\
\hline
$[[30, 28, 2]]$ & 131313131313131313131313131313 & 22 &[30, 1, 30]& [30, 29, 2] &$((30, 4^{1}$ $  2^{29}))$ \\
\bottomrule
\end{tabular*}
\end{minipage}
\end{center}
\end{table}

From the computational results, we made the following observation regarding quantum codes obtained via the binary codes
$C^{(1)}$ and $C^{(2)}$ associated with a quaternary code described in \cite{Wan1997}. First, we recall the definition of an MDS code. One of the
elementary bounds on the parameters $[n, k, d]$ of a classical linear code over a finite field is the Singleton bound which states that $k \leq n-d+1.$ Codes whose parameters attain equality in this bound are called MDS codes.

For a binary quantum code with parameters $[[n,k,d]]$,
the quantum Singleton bound is
\[
k \leq n-2d+2.
\]
A quantum code attaining equality in this bound is called
a quantum maximum-distance-separable (MDS) code.
In particular, the parameters $[[n,n-2,2]]$ satisfy
$n-2=n-2(2)+2$, so any quantum code with these
parameters attains the bound and is quantum MDS.

\begin{theorem}
Let $C$ be a $\mathbb{Z}_4$-linear code of even
length $n\geq 4$ whose binary residue and torsion codes
satisfy
\[
C^{(1)}=\langle \mathbf{1}_n\rangle
\quad\text{and}\quad
C^{(2)}=\{v\in\mathbb{F}_2^n:
v_1+\cdots+v_n=0\}.
\]
Then the CSS construction applied to
$C^{(1)}\subseteq C^{(2)}$ produces a binary
quantum MDS code with parameters $[[n,n-2,2]]$.
\end{theorem}

\begin{proof}
Since $n$ is even, the all-one vector
$\mathbf{1}_n$ has even Hamming weight. Therefore,
\[
C^{(1)}\subseteq C^{(2)}.
\]
The residue code is the binary repetition code
with parameters $[n,1,n]$, while the torsion
code is the binary even-weight code with
parameters $[n,n-1,2]$.

By the CSS construction stated in Section 4.6,
the resulting quantum code has dimension
\[
k=(n-1)-1=n-2.
\]
Moreover,
\[
(C^{(1)})^\perp=C^{(2)}
\quad\text{and}\quad
(C^{(2)})^\perp=C^{(1)}.
\]
Consequently, both set differences appearing
in the CSS distance formula are equal to
$C^{(2)}\setminus C^{(1)}$. For even $n\geq 4$,
this set contains vectors of Hamming weight
$2$ and no nonzero vectors of smaller weight.
Thus the quantum minimum distance is $d=2$.

The resulting code therefore has parameters
$[[n,n-2,2]]$. Since
\[
n-2=n-2(2)+2,
\]
it attains the quantum Singleton bound and
is a quantum MDS code.
\end{proof}

\noindent We also obtain some QECCs from negacyclic codes using CSS construction. These are summarized in Table 8.

In Table 8, the ``Binary Code Parameters''
column gives the parameters of the classical
binary code used in the CSS construction.
The last column records the additional properties
of that binary code, rather than properties
of the resulting quantum code.

\renewcommand{\arraystretch}{2}
\begin{table}[h]
\begin{center}
\begin{minipage}{\textwidth}
\caption{Quantum Codes from  negacyclic  codes of even length over $\ZZ$ }\label{tab3}
\begin{tabular*}{\textwidth}{@{\extracolsep{\fill}}| l|l|l|l|l|@{\extracolsep{\fill}}}
\toprule
$[[n,k,d]]$ & $f$ & Binary Code Parameters & References & Properties of Binary Code \\
\hline
$[[24,22, 2]]^{*}$ &13&[24,23, 2]&\cite{database}& Reversible and Dual Containing  \\
\hline
$[[24,20, 2]]$ &121&[24,22, 2]&\cite{database}& Reversible and Dual Containing  \\
\hline
$[[24,18, 2]]$ &1133&[24,21, 2]&\cite{database}& Reversible and Dual Containing  \\
\hline
$[[40,38, 2]]^{*}$ &13&[40,39, 2]&\cite{database}& Reversible and Dual Containing \\
\hline
$[[40,36, 2]]$ &121&[40,38,2]&\cite{database}& Reversible and Dual Containing \\
\hline
$[[40,34, 2]]$ &1133&[40,37,2]&\cite{database}& Reversible and Dual Containing \\
\hline

\end{tabular*}
\end{minipage}
\end{center}
*: best known quantum codes.
\end{table}

\subsection{Data and Code Availability}
\label{subsec:data-availability}

The computational data and Magma scripts supporting the results
presented in this paper are available in the GitHub repository
\url{https://github.com/Gidowu/Z4-Cyclic-Negacyclic-Codes}.
The repository provides access to computational resources associated
with the search for linear codes over $\mathbb{Z}_4$, particularly
cyclic and negacyclic codes of even length.

The repository is organized into three main directories:
\texttt{scripts}, \texttt{data}, and \texttt{documentation}.
The \texttt{scripts} directory contains programs used in the
computational investigations, including scripts for constructing
codes and examining their parameters. The \texttt{data} directory
contains CSV files recording computational results, including
code parameters and generator information. These files are
organized into subdirectories according to the types and lengths
of the codes considered. The \texttt{documentation} directory
provides additional information about the computational files
and their organization.

These resources complement the results summarized in this section
by providing access to the underlying computational data and
implementations. They also allow interested readers to examine
the recorded code parameters and explore the computational
methods used in this study.

\clearpage


\clearpage

\restoregeometry

\section{Concluding Remarks}
\label{sec:con}

By conducting computer searches to find new cyclic and negacyclic codes of even length over $\ZZ$, we have found 3270 new quaternary linear codes that have been added to the database \cite{Z4database}.

These computational results demonstrate the usefulness
of searching different families of cyclic and negacyclic
codes over $\mathbb{Z}_4$. To assess the significance
of the results, we distinguish codes that improve the
previously recorded minimum Lee distance from those
that attain an already known best value. We also
consider how the resulting codes are distributed
across the different code families and lengths.
 In some cases, there were no entries in the database with the given set of parameters, in other cases the codes we have found improved the previously best known parameters. Moreover, we used elementary methods to obtain binary quantum codes with best known parameters from  $\mathbb{Z}_4$ codes. However, none of these quantum codes leads to new codes with better parameters.  Although our computational searches covered several families of cyclic and negacyclic codes, they were limited by the size of the search spaces and the cost of computing minimum Lee distances. For oddly even cyclic codes, the search was exhaustive up to length 48, while 10,000 candidate codes were sampled at each larger length considered, up to length 98. For negacyclic codes, exhaustive searches were carried out up to length 30. Consequently, our results do not establish that the reported codes are optimal among all linear codes over $\mathbb{Z}_4$. Also, in the case of codes with the same parameters as previously best known quantum codes, we did not check equivalence between the existing codes and the ones we found.

Our work can be extended by finding ways to handle computer searches for cyclic and negacyclic codes of larger lengths, and exploring other generalizations of cyclic codes such as quasi-twisted (QT) codes, polycyclic codes, and generalized polycyclic codes of even lengths over $\ZZ$.

The main contribution of this work is the
computational investigation of classical cyclic
and negacyclic codes over $\mathbb{Z}_4$,
particularly their minimum Lee distances.
The quantum-code constructions provide an
additional application of these classical codes,
rather than improvements to existing quantum-code
records.

\section*{Acknowledgments}
We  thank the Kenyon Summer Science Scholars program for  funding this research.

\bibliographystyle{siamplain}
\bibliography{references}
\end{document}